\documentclass[11pt,letterpaper]{article}
\usepackage{amssymb,amsmath,amsthm}
\usepackage{graphicx,epsfig}
\usepackage{charter}
\newtheorem{theorem}{Theorem}

\newtheorem{proposition}{Proposition}
\newtheorem{corollary}{Corollary}
\newtheorem{lemma}{Lemma}

\newtheorem{definition}{Definition}

\newcommand{\tw}{{\mathbf{tw}}}

\newcommand{\pmin}{{\sc Min-CMSO}}
\newcommand{\pmax}{{\sc Max-CMSO}}
\newcommand{\pmm}{{\sc Min/Max-CMSO}}

\newcommand{\cO}{\mathcal{O}}

\newcommand{\h}[1]{\end{document}}

\usepackage{vmargin}

\setmarginsrb{1in}{1in}{1in}{1in}{0pt}{0pt}{0pt}{7mm}

\begin{document}

\title{Bidimensionality and EPTAS}
\date{}
\author{ Fedor V. Fomin\thanks{Department of Informatics, University of Bergen, Norway.
 \newline $~$\hspace{.5cm} 
\texttt{\{fedor.fomin|daniello\}@ii.uib.no}} 
~~~~\addtocounter{footnote}{-1}Daniel Lokshtanov\footnotemark 
~~~~Venkatesh Raman\thanks{The Institute of Mathematical Sciences, Chennai, India.
 \newline $~$\hspace{.5cm}
\texttt{\{vraman|saket\}@imsc.res.in}}~~~~
~~~~\addtocounter{footnote}{-1}Saket Saurabh\footnotemark~~~~
}

\maketitle


\begin{abstract}
\noindent
Bidimensionality theory appears to be a powerful framework for the development of meta-algorithmic techniques. It was introduced by Demaine et al. [{\em J. ACM 2005}\,] as a tool to obtain sub-exponential time parameterized algorithms for 
problems on $H$-minor free graphs. Demaine and Hajiaghayi [{\em SODA 2005}\,] extended the theory to obtain polynomial time approximation schemes (PTASs) for bidimensional problems, and subsequently improved these results to EPTASs. Fomin et. al [{\em SODA 2010}\,] established a third meta-algorithmic direction for bidimensionality theory by relating it to the existence of linear kernels for parameterized problems. In this paper we revisit bidimensionality theory from the perspective of approximation algorithms and redesign the framework for obtaining EPTASs to be more powerful, easier to apply and easier to understand.

Two of the most widely used approaches to obtain PTASs on planar graphs are the Lipton-Tarjan separator based approach [{\em SICOMP 1980}\,], and Baker's approach [{\em J.ACM 
  1994}\,]. Demaine and Hajiaghayi [{\em SODA 2005}\,] strengthened both approaches using bidimensionality and obtained EPTASs for several problems, including {\sc Connected Dominating Set} and {\sc Feedback Vertex Set}. 
 We unify the two strenghtened approaches to combine the best of both worlds. At the heart of our framework is a decomposition lemma which states
that for ``most'' bidimensional problems, there is a polynomial time algorithm which given an $H$-minor-free graph $G$ as input and an $\epsilon > 0$ outputs a vertex set $X$ of size $\epsilon \cdot OPT$ such that the treewidth of $G \setminus X$ is $f(\epsilon)$. Here, $OPT$ is the objective function value of the problem in question and $f$ is a function depending only on $\epsilon$. This allows us to obtain EPTASs on (apex)-minor-free graphs for all problems covered by the previous framework, as well as for a wide range of  packing problems, partial covering problems and problems that are neither closed under taking minors, nor contractions. To the best of our knowledge  for many of these problems including {\sc Cycle Packing}, {\sc Vertex-${\cal H}$-Packing}, {\sc Maximum Leaf Spanning Tree}, and {\sc Partial $r$-Dominating Set}  no EPTASs on planar graphs were previously known.
\end{abstract}


\section{Introduction}\label{sec:intro}
While most interesting graph problems remain {\cal NP} complete even when restricted to planar graphs, the restriction of a problem to planar graphs is usually considerably more tractable algorithmically than the problem on general graphs. Over the last four decades, it has been proved that many graph problems on planar graphs admit subexponential time algorithms~\cite{DornFT08-csr,FominThilJGT05,LiptonT80}, subexponential time parameterized algorithms~\cite{AlberBFKN02, FominT06}, (Efficient) Polynomial Time Approximation Schemes ((E)PTAS)~\cite{Baker94,Grohe:2003kt,Dawar:2006qb,Eppstein00,Gandhi2004,Khanna:1996rq} and linear kernels~\cite{AFN04,H.Bodlaender:2009ng,ChenFKX07}. Amazingly, the emerging theory of Bidimensionality developed by Demaine et al.~\cite{DemaineH08a-bi,Demaine:2008mi,DFHT05} is able to simultaneously explain the tractability of these problems within the paradigms of parameterized algorithms~\cite{DFHT05}, approximation~\cite{DH05} and kernelization~\cite{F.V.Fomin:2010oq}. The theory is built on  cornerstone theorems from Graph Minors Theory of Robertson and Seymour, and allows not only to explain the tractability of many problems, but also to generalize the results from planar graphs and graphs of bounded genus 
to graphs excluding a fixed minor. Roughly speaking, a problem is bidimensional if the solution value for the problem on a $k\times k$-grid is $\Omega(k^2)$, and the contraction or removal of an edge does not increase solution value. Many natural problems are bidimensional, including {\sc Dominating Set}, {\sc Feedback Vertex Set}, {\sc Edge Dominating Set}, {\sc Vertex Cover}, {\sc $r$-Dominating Set}, {\sc Connected Dominating Set}, {\sc Cycle Packing}, {\sc Connected Vertex Cover}, and {\sc  Graph Metric TSP}. 
 

A PTAS is an algorithm which takes an instance $I$ of an optimization problem and a parameter $\epsilon>0$ and, runs in time $n^{\cO(f(1/\epsilon)}$, produces a solution that is within a factor $\epsilon$ of being optimal. A PTAS with running time $f(1/\epsilon)\cdot n^{\cO(1)}$, is called efficient PTAS (EPTAS). Prior to bidimensionality~\cite{DH05}, there were two main approaches to design (E)PTASs on planar graphs. The first one was based on the classical Lipton-Tarjan planar separator theorem \cite{LiptonT79}. The second, more widely used approach was given by Baker~\cite{Baker94}. In the Lipton-Tarjan based approach we split the input $n$-vertex graph into two pieces of approximately equal size using a separator of size $O(\sqrt{n})$. Then we recursively approximate the problem on the two smaller instances and glue the approximate solutions at the separator. This approach was only applicable to problems where the size of the optimal solutions was at least a constant fraction of $n$. 

The main idea in Baker's approach is to decompose the planar graph into overlapping subgraphs of bounded outerplanarity and then solve the problem optimally in each of these subgraphs using dynamic programming. Later Eppstein~\cite{Eppstein00} generalized this approach to work for larger class of graphs, namely apex minor free graphs. Khanna and Motwani~\cite{Khanna:1996rq} used Baker's approach in an attempt to syntactically characterize the complexity class of  problems admitting PTASs, establishing a family of problems on planar graphs to which it applies. The same kind of approach is also used by Dawar et al.~\cite{Dawar:2006qb} to obtain EPTASs for every minimization problem definable in first-order logic on every class of graphs excluding a fixed minor. Baker's approach seemed to be limited to ``local'' graph problems--where one is interested in finding a vertex/edge set satisfying a property that can be checked by looking at constant size neighborhood around each vertex. 

Demaine and Hajiaghayi~\cite{DH05} used bidimensionality theory to strengthen and generalize both approaches. In particular they strengthened the Lipton-Tarjan approach significantly by showing that for a magnitude of problems one can find a separator of size $O(\sqrt{OPT})$ that splits the optimum solution evenly into two pieces. Here $OPT$ is the size of the optimum solution. This allowed them to give EPTASs for several problems on planar graphs and more generally on apex-minor-free graphs or $H$-minor free graphs. Two important problems to which their approach applies are {\sc Feedback vertex Set} and {\sc Connected Dominating Set}. Earlier only a PTAS and no EPTAS for {\sc Feedback Vertex Set} on planar graphs was known~\cite{KleinbergK01}. In addition, they also generalize Baker's approach by allowing more interaction between the overlapping subgraphs.

Comparing the generalized versions of the two approaches, it seems that each has its strengths and weaknesses. In the generalized Lipton-Tarjan approach of Demaine and Hajiaghayi~\cite{DH05} one splits the graph into two pieces recursively. To ensure that the repeated application does not ``increase'' the approximation factor, in each recursive step one needs to carefully reconstruct the solution from the smaller ones. Additionally, to ensure that the separator splits the optimum solution evenly, the framework of Demaine and Hajiaghayi~\cite{DH05} requires a constant factor approximation for the problem in question. On the other hand, their generalization of Baker's approach essentially identifies a set of vertices or edges that interacts in a limited way with the optimum solution, such that the removal of $X$ from the input graph leaves a graph on which the problem can be solved optimally in polynomial time. The set $X$ could be as large as $O(n)$ which in some cases makes it difficult to bound the amount of interaction between the set $X$ and the optimum solution.

In this paper we propose a framework which combines the best of both worlds---the generalized Lipton-Tarjan and generalized Baker's approaches. In particular, we show 
that for most bidimensional problems there is a polynomial time algorithm that given a graph $G$ and an $\epsilon > 0$ outputs a vertex set $X$ of size $\epsilon \cdot OPT$ such that the treewidth of $G \setminus X$ is $f(\epsilon)$. Because the {\em size} of $X$ is bounded, the interaction between $X$ and the optimum solution is bounded trivially. Since $X$ is only removed once, the difficulty faced by a recursive approach vanishes. In our framework to obtain EPTASs, we demand that the problem in question is ``reducible'', which is nothing else than that the set $X$ can be removed from the graph, disturbing the optimum solution by at most $O(\epsilon \cdot OPT)$. Finally, our algorithm to compute $X$ does not require an approximation algorithm for the problem in question, and relies only on a sublinear treewidth bound. For most problems, such a bound can be obtained via bidimensionality, whereas for some problems that are not bidimensional, one can obtain the sublinear treewidth bound directly.



Our new framework allows   to obtain EPTASs on (apex)-minor-free graphs for all problems covered by the previous framework, as well as for several packing problems, partial covering problems and problems that are neither closed under taking minors nor contractions. For an example consider the {\sc Maximum Degree Preserving Spanning Tree} problem where given a graph $G$ the objective is to find a spanning tree such that the number of vertices which have the same degree in the tree as in the input graph is maximized.
 {\sc Maximum Degree Preserving Spanning Tree} is neither closed under taking minors nor contractions, but one can still 
apply our framework to obtain an EPTAS for this problem. For another example, consider {\sc Cycle Packing}, where given a graph $G$ the objective is to find the maximum number of vertex disjoint cycles in $G$. For this problem, it is not clear how to directly apply the previous framework to obtain an EPTAS. On the other hand, using our framework to obtain an EPTAS for this problem is easy. More generally, we give an EPTAS for the {\sc Vertex-${\cal H}$-Packing} problem, defined as follows.
Let ${\cal H}$ be a finite set of connected graphs such that at least one graph in ${\cal H}$ is planar. 
Input to {\sc Vertex-${\cal H}$-Packing} is a graph $G$ and the objective is to find a maximum size collection of vertex disjoint subgraphs $G_{1},\ldots,G_{k}$ of $G$ such that each of them contains some graph from ${\cal H}$ as a minor. To the best of our knowledge no EPTASs for {\sc Cycle Packing}, {\sc Vertex-${\cal H}$-Packing}, {\sc Maximum Leaf Spanning Tree}, or {\sc Partial $r$-Dominating Set} were previously known, even on planar graphs. Our framework to obtain EPTASs seems to be the most general one could hope for in the context of bidimensionality and approximation.

\section{Definitions and Notations}



%




%

In this section we give various definitions which we make use of in the paper. Let~$G$ be a graph with vertex set $V(G)$ and edge set $E(G)$. A graph~$G'$ is a
 \emph{subgraph} of~$G$ if~$V(G') \subseteq V(G)$ and~$E(G') \subseteq E(G)$. The subgraph~$G'$ is called an \emph{induced subgraph} of~$G$ if~$E(G')
 = \{ uv \in E(G) \mid u,v \in V(G')\}$, in this case, $G'$~is also called the subgraph \emph{induced by~$V'$} and denoted by~$G[V']$. For a vertex set $S$, by $G \setminus S$ we denote $G[V(G) \setminus S]$. A graph class ${\cal G}$ is {\em hereditary} if for any graph $G \in {\cal G}$ all induced subgraphs of $G$ are in ${\cal G}$.
  By $N(u)$ we denote (open) neighborhood of $u$, that is, the set of all vertices adjacent to $u$. Similarly, by $N[u]$ we denote $N(u) \cup \{u\}$.  
For a subset $D \subseteq V$, we define $N[D]=\cup_{v\in D} N[v]$ and $N(D) = N[D] \setminus D$. 
The {\it distance} $d_G(u,v)$ between two vertices $u$ and $v$ of $G$ is the length of the shortest path in $G$ from $u$ to $v$. Define $B_G^r(v)$ to be the set of vertices within distance at most $r$ from $v$, including $v$ itself. For a vertex set $S$ define $B_G^r(v) = \bigcup_{v \in S} B_G^r(v)$.


%

\paragraph{Minors.}
Given an edge  $e=xy$ of a graph $G$, the graph  $G/e$ is obtained from  $G$ by contracting the edge $e$. That is, the endpoints $x$  and $y$ are replaced by a new vertex $v_{xy}$
which  is  adjacent to the old neighbors of $x$ and $y$ (except from $x$ and $y$).  A graph $H$ obtained by a sequence of edge-contractions is said to be a \emph{contraction} of $G$.  We denote it by $H\leq_{c} G$. A graph $H$ is a {\em minor} of a graph $G$ if $H$ is the contraction of some subgraph of $G$ and we denote it by $H\leq_{m} G$. We say that a graph $G$ is {\em $H$-minor-free} if $G$ does not contain $H$ as a minor. We also say that a graph class ${\cal G}_H$ is {\em $H$-minor-free} (or, excludes $H$ as a minor) when all its members are $H$-minor-free. An \emph{apex graph} is a graph obtained from a planar graph $G$ by adding a vertex and making it adjacent to some of the vertices of $G$. A graph class ${\cal G}_H$ is \emph{apex-minor-free} if ${\cal G}_H$ excludes a fixed apex graph $H$ as a minor. Let us remark that every planar, and more generally, graph of bounded genus, is an   $H$-minor-free graph for some fixed apex graph $H$.

\paragraph{\bf Grids and their triangulations.}
Let  $r$ be a positive integer, $r\geq 2$. The
\emph{$(r\times r)$-grid} is the Cartesian product of two paths of
lengths $r-1$.  A vertex of a  grid is a {\em corner} if it has degree $2$. Thus each
$(r\times r)$-grid has 4 corners. A vertex of a  $(r\times
r)$-grid is called {\em internal} if it has degree 4, otherwise it
is called {\em external}.
Let $\Gamma_{r}$ be the graph obtained from the  $(r\times r)$-grid by
triangulating internal faces of the $(r\times r)$-grid such that all internal vertices become  of degree $6$,
all non-corner external vertices are of degree 4,
and then one corner of degree two is joined by edges with all vertices
of the external face. The graph $\Gamma_6$ is shown in Fig.~\ref{fig-gamma-k}.
\begin{figure}[ht]
\begin{center}
\scalebox{.69}{\includegraphics{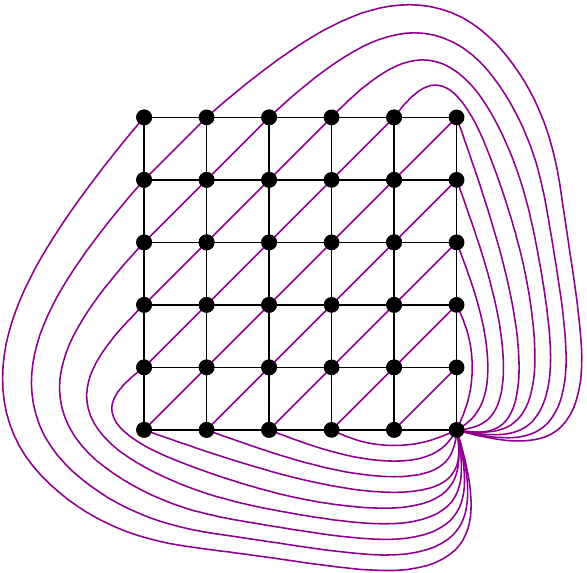}}
\end{center}
\caption{The graph $\Gamma_{6}$.}
\label{fig-gamma-k}
\end{figure}

\paragraph{\bf Treewidth.}
A \emph{tree decomposition} of a graph $G$ is a pair $(\mathcal{X},T)$, where $T$
is a tree and ${\cal X}=\{X_{i} \mid i\in V(T)\}$ is a collection of subsets
of $V$ such that the following conditions are satisfied.
\begin{enumerate}\setlength\itemsep{-1.2mm}
\item $\bigcup_{i \in V(T)} X_{i} = V(G)$.
\item For each edge $xy \in (G)$, $\{x,y\}\subseteq X_i$ for some $i\in V(T)$.
\item For each $x\in V(G)$ the set $\{ i \mid x \in X_{i} \}$ induces a connected subtree of $T$.
\end{enumerate}
The \emph{width} of the  tree decomposition is $\max_{i \in V(T)}\,|X_{i}| - 1$. The \emph{treewidth} of a
graph $G$, ${\bf tw}(G)$, is the minimum width over all tree decompositions of $G$.
%

\paragraph{Counting Monadic Second Order Logic.}
\label{countmsop}
The syntax of MSO of graphs includes the logical connectives $\vee$, $\land$, $\neg$, $\Leftrightarrow $,  $\Rightarrow$, variables for vertices, edges, set of vertices and set of edges, the quantifiers $\forall$, $\exists$ that can be applied to these variables, and the following five binary relations: 
\begin{enumerate}\setlength\itemsep{-1.2mm}
\item $u\in U$ where $u$ is a vertex variable and $U$ is a vertex set variable.
\item $d \in D$ where $d$ is an edge variable and $D$ is an edge set variable.
\item $\mathbf{inc}(d,u)$, where $d$ is an edge variable,  $u$ is a vertex variable, and the interpretation is that the edge $d$ is incident on the vertex $u$.
\item $\mathbf{adj}(u,v)$, where  $u$ and $v$ are vertex variables $u$, and the interpretation is that $u$ and $v$ are adjacent.
\item Equality of variables, $=$, representing vertices, edges, set of vertices and set of edges.
\end{enumerate}
CMSO, or {\em counting monadic second-order logic} is MSO additionally equipped with an atomic formula $\mathbf{card}_{n,p}(U)$ for testing whether the cardinality of a set $U$ is congruent to $n$ modulo $p$, where $n$ and $p$ are integers independent of the input graph such that $0\leq n<p$ and $p\geq 2$.
We refer to ~\cite{ArnborgLS91,Courcelle90,Courcelle97} for a detailed introduction to CMSO. 

\pmin{} and \pmax{} problems are graph optimization problems where the objective is to find a maximum or minimum sized vertex or edge set satisfying a CMSO-expressible property. 
In particular, a \pmm{} graph problem $\Pi$ we are given a graph $G$ as input. The objective is to find a minimum/maximum cardinality vertex vertex/edge set $S$ such that the CMSO-expressible predicate $P_\Pi(G,S)$ is satisfied.
%
%

\paragraph{Bidimensionality and Separability.}
Our results concern graph optimization problems where the objective is to find a vertex or edge set that satisfies a feasibility constraint and maximizes or minimizes a problem-specific objective function. For a problem $\Pi$ and vertex (edge) set $S$ let $\phi_\Pi(G,S)$ be the feasibility constraint returning {\bf true} if $S$ is feasible and {\bf false} otherwise. Let $\kappa_\Pi(G,S)$ be the objective function. 
In most cases, $\kappa_\Pi(G,S)$ will return $|S|$. We will only consider problems where every instance has at least one feasible solution. Let ${\cal U}$ be the set of all graphs. 
For a graph optimization problem $\Pi$ let $\pi : {\cal U} \rightarrow \mathbb{N}$ be a function returning the objective function value of the optimal solution of $\Pi$ on $G$. We say that a problem $\Pi$ is {\em minor-closed} if $\pi(H) \leq \pi(G)$ whenever $H$ is a minor of $G$. Similarly, we say $\Pi$ is {\em contraction-closed} if $\pi(H) \leq \pi(G)$ whenever $H$ is a contraction of $G$. We now define bidimensional problems.

\begin{definition}[\cite{DFHT05,FominGTesa09}]
  A graph optimization problem $\Pi$ is minor-bidimensional if 
  \begin{enumerate}
  \setlength{\itemsep}{-1.2mm}
  \item $\Pi$ is minor-closed,
  \item there is a $\delta > 0$ such that $\pi(R) \geq \delta r^2$ for the $(r \times r)$-grid $R$. In other words, the value of the solution on $R$ should be at least   $ \delta r^2$.
  \end{enumerate}
  \noindent
  A graph optimization problem $\Pi$ is called \emph{contraction-bidimensional} if 
 \begin{enumerate}
 \setlength{\itemsep}{-1.2mm}
  \item $\Pi$ is contraction-closed,
  \item there is $\delta > 0$ such that $\pi(\Gamma_r) \geq \delta r^2$. 
  \end{enumerate}
  In either case, $\Pi$ is called \emph{bidimensional}.  
\end{definition}

Demaine and Hajiaghayi~\cite{DH05} define the  \emph{separation} property for problems, and show how separability together with bidimensionality is useful to obtain EPTASs on $H$-minor-free graphs. In our setting a slightly weaker notion of separability is sufficient. In particular the following definition is a reformulation of the requirement $3$ of the definition of separability in~\cite{DH05} and similar to the definition used in~\cite{F.V.Fomin:2010oq} to obtain kernels for bidimensional problems. 

\begin{definition} \label{def:sep}
A minor-bidimensional problem $\Pi$ has the \emph{separation} property if given any graph $G$ and a partition of $V(G)$ into $L \uplus S \uplus R$ such that $N(L) \subseteq S$ and $N(R) \subseteq S$, and given an optimal solution $OPT$ to $G$, $\pi(G[L]) \leq \kappa_\Pi(G[L], OPT \cap L) + \cO(|S|)$ and $\pi(G[R]) \leq \kappa_\Pi(G[R], OPT \cap R) + \cO(|S|)$.
\end{definition}
For contraction-bidimensional parameters we have a slightly different definition of the separation property. For a graph $G$ and a partition of $V(G)$ into $L \uplus S \uplus R$ such that $N(L) \subseteq S$ and $N(R) \subseteq S$, we define $G_L$ ($G_R$ ) to be the graph obtained from $G$ by contracting every connected component of $G[R]$ ($G[L]$) into the vertex in $S$ with  the lowest index.  
\begin{definition}  \label{def:contrasep}
A contraction-bidimensional problem has the \emph{separation} property if given any graph $G$ and a partition of $V(G)$ into $L \uplus S \uplus R$ such that $N(L) \subseteq S$ and $N(R) \subseteq S$, and given an optimal solution $OPT$ to $G$, $\pi(G_L) \leq \kappa_\Pi(G_L, OPT \setminus R) + \cO(|S|)$ and $\pi(G_R) \leq \kappa_\Pi(G_R, OPT \setminus L) + \cO(|S|)$. 
\end{definition}

In Definitions~\ref{def:sep} and~\ref{def:contrasep} we slightly misused notation. Specifically, in the case that $OPT$ is an edge set we should not be considering $OPT \setminus R$ and $OPT \setminus L$ but $OPT \setminus E(G[R])$ and $OPT \setminus E(G[L])$ respectively.

\paragraph{Reducibility, $\eta$-Transversability and Graph Classes with Truly Sublinear Treewidth.}
We now define three of the central notions of this article.
 
\begin{definition}
A graph optimization problem $\Pi$ with objective function $\kappa_\Pi$ is called {\em reducible} if there exists a \pmm{} problem $\Pi'$ and a function $f : \mathbb{N} \rightarrow \mathbb{N}$ such that 
\begin{enumerate}\setlength\itemsep{-1.2mm}
 \item there is a polynomial time algorithm that given $G$ and $X \subseteq V(G)$ outputs $G'$ such that $\pi'(G') = \pi(G) \pm O(|X|)$ and $\tw(G') \leq f(\tw(G \setminus X))$,
 \item there is a polynomial time algorithm that given $G$ and $X \subseteq V(G)$, $G'$ and a vertex (edge) set $S'$ such that $P_{\Pi'}(G',S')$ holds, outputs $S$ such that $\phi_\Pi(G,S)=\textbf{true}$   and $\kappa_\Pi(G,S)=|S'| \pm O(|X|)$.
\end{enumerate}
\end{definition}

\begin{definition}
A graph optimization problem $\Pi$ is called $\eta$-Transversable if there is a polynomial time algorithm that given a graph $G$ outputs a set $X$ of size $O(\pi(G))$ such that $\tw(G \setminus X) \leq \eta$.
\end{definition}



\begin{definition}  
Graph class ${\cal G}$ has {\em truly sublinear treewidth with parameter}  $\lambda$, $ 0<\lambda<1$, if for every  $\eta>0$, there exists  
$\beta>0$ such that for any graph $G \in {\cal G}$ and $X \subseteq V(G)$ the condition  $\tw(G \setminus X) \leq \eta$ yields that 
$\tw(G) \leq \eta + \beta|X|^\lambda$. \end{definition}

\section{Partitioning Graphs of Truly Sublinear Treewidth}
We need the following well known lemma, see e.g. \cite{Bodlaender98},  on separators in graphs of bounded treewidth.  
\begin{lemma}
\label{lemma:balsep1}
Let $G$ be a graph of treewidth at most $t$ and $w : V(G) \rightarrow \mathbb{R}^{+}\cup \{0\}$  be a weight function. Then there is a partition of $V(G)$ into $L \uplus S \uplus R$ such that
\begin{itemize}\setlength\itemsep{-1.2mm}
 \item $|S| \leq t+1$, $N(L) \subseteq S$ and $N(R) \subseteq S$,
 \item every connected component $G[C]$ of $G \setminus S$ has $w(C)\leq w(V)/2$,
 \item $\frac{w(V(G))-w(S)}{3}\leq w(L)\leq \frac{2(w(V(G))-w(S))}{3}$ and $\frac{w(V(G))-w(S)}{3}\leq w(T)\leq \frac{2(w(V(G)-w(S))}{3}$.
\end{itemize}
\end{lemma}
The next lemma is crucial in our analysis.

\begin{lemma}\label{lem:theLemma} Let ${\cal G}$ be a hereditary graph class of truly sublinear treewidth with parameter $\lambda$. For any $\epsilon < 1$ there is a $\gamma$ such that for any $G \in {\cal G}$ and $X \subseteq V(G)$ with $\tw(G \setminus X) \leq \eta$ there is a $X' \subseteq V(G)$ satisfying $|X'| \leq \epsilon|X|$ and for every connected component $C$ of $G \setminus X'$ we have $|C \cap X| \leq \gamma$ and
$|N(C)| \leq \gamma$.
Moreover $X'$ can be computed from $G$ and $X$ in polynomial time, where the polynomial is independent of  $\epsilon$, $\lambda$ and $\eta$. 
\end{lemma}

\begin{proof}
For any $\gamma \geq 1$, define $T_\gamma : \mathbb{N} \rightarrow \mathbb{N}$ such that $T_\gamma(k)$ is the smallest integer such that if $G \in {\cal G}$ and there is a $X \subseteq V(G)$ with $\tw(G \setminus X) \leq \eta$ and $|X| \leq k$, then there is a $X' \subseteq V(G)$ of size at most $T_\gamma(k)$ such that for every connected component $C$ of $G \setminus X'$ we have $|C \cap X| \leq \gamma$ and $|N(C)| \leq \gamma$.
Informally, $T_\gamma(k)$ is the minimum size of a vertex set $X'$ such that every connected component $C$ of $G \setminus X'$ has at most $\gamma$ neighbours in $X'$ and contains at most $\gamma$ vertices of $X$, if we know that deleting the $k$-sized vertex set $X$ from $G$ yields a graph of treewidth $\eta$. Furthermore, since ${\cal G}$ is a hereditary graph class of truly sublinear treewidth with parameter $\lambda$ and 
$\tw(G \setminus X) \leq \eta$, there exists a constant $\beta$ such that $\tw(G) \leq \eta + \beta|X|^\lambda$. 
We will make choices for the constants $\delta$, $\gamma$ and $\rho$ based on $\eta$, $\lambda$, $\beta$ and $\epsilon$. Our aim is to show that $T_\gamma(k) \leq \epsilon k$ for every $k$.

Observe that for any numbers $a>0$, $b>0$, we have $a^\lambda + b^\lambda > (a+b)^\lambda$ since $\lambda < 1$. Thus we have $\rho = \min_{1/3 \leq \alpha \leq 2/3}\alpha^\lambda + (1-\alpha)^\lambda > 1$. We choose $\delta = \frac{(2\epsilon+1)(\beta+\eta+1))}{\rho - 1}$ and $\gamma=(\frac{3\delta}{\epsilon})^{\frac{1}{1-\lambda}}$. 
If $\tw(G \setminus X) \leq \eta$ and $|X| \leq \gamma$ then we set $X'=\emptyset$, so $T_\gamma(k) = 0 \leq \epsilon k$ for $k \leq \gamma$. We now show that if $k \geq \gamma / 3$ then $T_\gamma(k) = 0 \leq \epsilon k - \delta k^\lambda$ by induction on $k$. For the base case if $\gamma / 3 \leq k \leq \gamma$ then the choice of $\gamma$ implies the following inequality.
$$\epsilon k - \delta k^\lambda \geq \epsilon \frac{\gamma}{3} - \delta \gamma^\lambda \geq 0 = T_\gamma(k)$$

We now consider $T_\gamma(k)$ for $k > \gamma$. Let $G \in {\cal G}$ and $X \subseteq V(G)$ with $\tw(G \setminus X) \leq \eta$ and $|X| \leq k$. The treewidth of $G$ is at most $\eta + \beta k^\lambda$. Construct a weight function $w : V(G) \rightarrow \mathbb{N}$ such that $w(v)=1$, when $v \in X$ and $w(v)=0$ otherwise. By Lemma~\ref{lemma:balsep1}, there is a partition of $V(G)$ into $L$, $S$ and $R$ such that $|S| \leq \eta + \beta k^\lambda + 1$, $N(L) \subseteq S$, $N(R) \subseteq S$, $|L \cap X| \leq 2k/3$ and $|R \cap X| \leq 2k/3$. Deleting $S$ from the graph $G$ yields two graphs $G[L]$ and $G[R]$ with no edges between them. Thus we put $S$ into $X'$ and then proceed recursively in $G[L \cup S]$ and $G[R \cup S]$ starting from the sets $S \cup (X \cap L)$ and $S \cup (X \cap R)$ in $G[L \cup S]$ and $G[R \cup S]$ respectively. This yields the following recurrence for $T_\gamma$.
$$
T_\gamma(k) \leq \max_{1/3 \leq \alpha \leq 2/3} T(\alpha k + \eta + \beta k^\lambda + 1) + T((1-\alpha) k + \eta + \beta k^\lambda + 1) + \eta + \beta k^\lambda + 1.
$$
Observe that since $k \geq \gamma$ we have $\alpha k \geq \gamma / 3$ and $(1 - \alpha k) \geq \gamma / 3$. The induction hypothesis then yields the following inequality.
%
\begin{align*}
T_\gamma(k) & \leq & \max_{1/3 \leq \alpha \leq 2/3} & T(\alpha k + \eta + \beta k^\lambda + 1) + T((1-\alpha) k + \eta + \beta k^\lambda + 1) + \eta + \beta k^\lambda + 1 \\
& \leq & \max_{1/3 \leq \alpha \leq 2/3} & \epsilon k - \delta (\alpha k)^\lambda - \delta((1-\alpha)k)^\lambda + (2\epsilon+1)(\beta k^\lambda + \eta + 1) \\
& \leq & \max_{1/3 \leq \alpha \leq 2/3} & \epsilon k - \delta k^\lambda(\alpha^\lambda + (1-\alpha)^\lambda) + (2\epsilon+1)(\beta k^\lambda + \eta + 1) \\
& \leq & & \epsilon k - \delta k^\lambda - \delta (\rho-1) k^\lambda + (2\epsilon +1)(\beta k^\lambda + \eta + 1) \\
& \leq & & \epsilon k - \delta k^\lambda.
\end{align*}
%
The last inequality holds whenever $\delta (\rho-1) k^\lambda \geq (2\epsilon +1)(\beta k^\lambda + \eta + 1)$, which is ensured by the choice of $\delta$ and the fact that $k^\lambda \geq 1$. Thus $T_\gamma(k) \leq \epsilon k$ for all $k$. Hence there exists a set $X'$ of size at most $\epsilon k$ such that for every component $C$ of $G \setminus X'$ we have $C \cap X \leq \gamma$ and $|N(C)| \leq \gamma$. 


What remains is to show that $X'$ can be computed from $G$ and $X$ in polynomial time. The inductive proof can be converted into a recursive algorithm. The only computationally hard step of the proof is the construction of a tree-decompositon of $G$ in each inductive step. Instead of computing the treewidth exactly we use the $d^*\sqrt{\log \tw(G)}$-approximation algorithm by Feige et al.~\cite{FeigeHajLee05}, where $d^*$ is a fixed constant. Thus when we partition $V(G)$ into $L$, $S$, and $R$ using Lemma~\ref{lemma:balsep1}, the upper bound on the size of $S$ will be $d^*(\eta + \beta k^\lambda)\sqrt{\log(\eta + \beta k^\lambda)}$ instead of $\eta + \beta k^\lambda$. However, for any $\lambda < \lambda' < 1$ there is a $\beta'$ such that  $d^*(\eta + \beta k^\lambda)\sqrt{\log(\eta + \beta k^\lambda)} < \eta + \beta' k^{\lambda'}$. Now we can apply the above analysis with $\beta'$ instead of $\beta$ and $\lambda'$ instead of $\lambda$ to bound the size of the set $X'$ output by the algorithm. This concludes the proof of the lemma.
\end{proof}

The following corollary is a direct consequence of Lemma~\ref{lem:theLemma}. Nevertheless, we find it worthwhile to mention it separately.
\begin{corollary}\label{cor:theCorollary} Let ${\cal G}$ be a hereditary graph class of truly sublinear treewidth   with parameter $\lambda$. For any $\epsilon > 1$ there is a $\tau$ with $\tau = O((\frac{1}{\epsilon})^\frac{\lambda}{1-\lambda})$ such that for any $G \in {\cal G}$ and $X \subseteq V(G)$ with $\tw(G \setminus X) \leq \eta$ there is a $X' \subseteq V(G)$ satisfying $|X'| \leq \epsilon|X|$ such that $\tw(G \setminus X') \leq \tau$.
\end{corollary}

\begin{proof}
We apply Lemma~\ref{lem:theLemma} on $G$ and $X$ to obtain the set $X'$ of size $\epsilon |X|$. Observe that in the proof of  Lemma~\ref{lem:theLemma}, $\gamma = O((\frac{1}{\epsilon})^\frac{1}{1-\lambda})$. The treewidth of $G \setminus X'$ equals the maximum treewidth of a connected component $C$ of $G \setminus X'$. However $|C \cap X| \leq \gamma$ and so $\tw(G[C]) = O(\gamma^\lambda)$, concluding the proof.
\end{proof}

\section{Approximation Schemes}

\paragraph{Approximation Schemes for $\eta$-Transversable problems.}

\begin{theorem}\label{thm:eptastrans}
Let $\Pi$ be an $\eta$-transversable, reducible graph optimization problem. Then $\Pi$ has an EPTAS on every graph class ${\cal G}$ with truly sublinear treewidth.
\end{theorem}

\begin{proof}
Let $G$ be the input to $\Pi$ and $\epsilon > 0$ be fixed. Since $\Pi$ is $\eta$-transversable there is a polynomial time algorithm that outputs a set $X$ such that $|X| \leq \rho_1 \pi(G)$ and $\tw(G) \leq \eta$, for a fixed constant $\rho_1$.  Furthermore, since ${\cal G}$ is a hereditary graph 
class of truly sublinear treewidth with parameter $\lambda$ and $\tw(G \setminus X) \leq \eta$, there exists a constant $\beta$ such that 
$\tw(G) \leq \eta + \beta|X|^\lambda$.  Let $\epsilon'$ be a constant to be selected later. By Lemma~\ref{lem:theLemma}, there exist $\gamma$, $\lambda' < 1$ and $\beta'$ depending on $\epsilon'$, $\lambda$, $\eta$ and $\beta$ such that given $G$ and $X$ a set $X'$ with the following properties can be found in polynomial time. First $|X'| \leq \epsilon' |X|$, and secondly for every component $C$ of $G \setminus X'$ we have that $C \cap X \leq \gamma$. Thus $\tw(G \setminus X') = \tau \leq \beta' \gamma^{\lambda'} + \eta$. Since $\Pi$ is reducible there exists  a \pmm{} problem $\Pi'$, a constant $\rho_2$ and a function $f : \mathbb{N} \rightarrow \mathbb{N}$ such that: 
\begin{enumerate}\setlength\itemsep{-1.2mm}
 \item there is a polynomial time algorithm that given $G$ and $X' \subseteq V(G)$ outputs $G'$ such that $|\pi'(G') - \pi(G)| \leq \rho_2 |X'|$ and $\tw(G') \leq f(\tau)$,
 \item there is a polynomial time algorithm that given $G$ and $X' \subseteq V(G)$, $G'$ and a vertex (edge) set $S'$ such that $P_{\Pi'}(G',S')$ holds outputs $S$ such that $\phi_\Pi(G,S)$ holds and $|\kappa_\Pi(G,S)-|S'|| \leq \rho_2 |X'|$.
\end{enumerate}
We constuct $G'$ from $G$ and $X'$ using the first polynomial time algorithm. Since $\tw(G') \leq f(\tau)$ we can use an extended version of Courcelle's theorem~\cite{Courcelle90,Courcelle97}  given by  Borie et al.~\cite{BoriePT92} to find an optimal solution $S'$ to $\Pi'$ in $g(\epsilon')|V(G')|$ time. By the properties of the first polynomial time algorithm, $||S'|-\pi(G)| \leq \rho |X'|$ where $\rho = \max(\rho_1,\rho_2)$. We now use the second polynomial time algorithm to construct a  solution $S$ to $\Pi$ from $G$, $X'$, $G'$ and $S'$. The properties of the second algorithm ensure $\phi_\Pi(G,S)$ holds and that $|\kappa_\Pi(G,S)-|S'|| \leq \rho |X'|$, and hence $|\kappa_\Pi(G,S)-\pi(G)| \leq 2 \rho |X'| \leq 2 \rho^2 \epsilon' \pi(G)$. Choosing $\epsilon' = \frac{\epsilon}{2\rho^2}$ yields $|\kappa_\Pi(G,S)-\pi(G)| \leq \epsilon \pi(G)$, proving the theorem.
\end{proof}

\paragraph{Approximation Schemes for Bidimensional problems.} Now we use Theorem~\ref{thm:eptastrans} to give EPTASs for bidimensional, separable and reducible problems on graphs excluding a fixed $H$ as a minor. In order to do this, we show that $H$-minor-free graphs have truly sublinear treewidth and that bidimensional and separable problems are $\eta$-transversable. To show that $H$-minor free graphs have truly sublinear treewidth we use the following result.
\begin{proposition}[\cite{DemaineFH05,Demaine:2008dq,FominGTesa09}]\label{prop:lineagrid}
Let $G$ be a connected graph excluding a fixed graph $H$ as a minor. Then there exists some constant $c$ such that if $\tw(G)\geq c\cdot r^{2}$, then $G$ contains  the $r\times r$-grid as a minor. Moreover, if $H$ is an apex graph, then $G$ does not contain $\Gamma_{r}$ as a contraction.
\end{proposition}

\begin{corollary}\label{cor:minortreewidth}
Let ${\cal G}_H$ be a class of graphs excluding a fixed graph $H$ as a minor. Then ${\cal G}_H$ has truly sublinear treewidth with $\lambda = \frac{1}{2}$.
\end{corollary}

\begin{proof}
Let $\rho$ be a constant such that any graph $G \in {\cal G}_H$ of treewidth at least $t$ contains a $\rho t \times \rho t$ grid as a minor. Let $G \in {\cal G}_H$ have a vertex subset $X$ such that $\tw(G \setminus X) \leq \eta$ for a fixed constant $\eta$. We prove that $\tw(G) \leq \frac{\eta+1}{\rho}\lceil \sqrt{|X|+1} \rceil$. Suppose not. Then, by Proposition~\ref{prop:lineagrid}, $G$ contains a $(\eta+1)\lceil \sqrt{|X|+1} \rceil \times (\eta+1)\lceil \sqrt{|X|+1} \rceil$ grid as a minor. Hence $G$ contains at least $|X|+1$ disjoint $\eta+1 \times \eta+1$ grids as a minor. The set $X$ is disjoint from at least one of these grids, and hence $G \setminus X$ contains a  $\eta+1 \times \eta+1$ grid as a minor and has treewidth at least $\eta+1$, yielding the desired contradiction.
\end{proof}

For every fixed integer $\eta$ we define the {\sc $\eta$-Transversal} problem as follows. Input is a graph $G$, and the objective is to find a minimum cardinality vertex set $S \subseteq V(G)$ such that $\tw(G \setminus S) \leq \eta$. We now give a polynomial time constant factor approximation for the  {\sc $\eta$-Transversal} problem on  on $H$-minor-free graphs.

\begin{lemma}\label{lem:transversapprox} 
For every integer $\eta$ and fixed graph $H$ there is a constant $c$ and a polynomial time $c$-approximation algorithm for the {\sc $\eta$-Transversal} problem on $H$-minor-free graphs. The polynomial in the running time only depends on $H$ and $\eta$. 
\end{lemma}

\begin{proof}
Let $X$ be a smallest vertex set in $G$ such that $\tw(G \setminus X) \leq \eta$. By Lemma~\ref{lem:theLemma} with $\epsilon = 1/2$ there exists a $\gamma$ depending only on $H$ and $\eta$ and a set $X'$ with $|X'| \leq |X| / 2$ such that for any component $C$ of $G \setminus X'$, $|C \cap X| \leq \gamma$ and $|N(C)| \leq \gamma$. Since $X$ is the {\em smallest} set such that $\tw(G \setminus X) \leq \eta$, there is a component $C$ of $G \setminus X'$ with treewidth strictly more than $\eta$. Let $Z = N(C)$ and observe that $Z \subseteq X'$ is a set of size at most $\gamma$ such that $C$ is a conneced component of $G \setminus Z$.

The algorithm proceeds as follows. It tries all possibilities for $Z$ and looks for a connected component $C$ of $G \setminus Z$ such that $\eta < \tw(G[C]) = O(\sqrt{\gamma})$. It solves the {\sc $\eta$-Transversal} problem optimally on $G[C]$ by noting that {\sc $\eta$-Transversal} can be formulated as a \pmin{} problem and applying the algorithm by Borie et al~\cite{BoriePT92}. Let $X_C$ be the solution obtained for $G[C]$. The algorithm adds $X_C$ and $N(C)$ to the solution and repeats this step on $G \setminus (C \cup N(C))$ as long as $\tw(G) \geq \eta$. 

Clearly, the set returned by the algorithm is a feasible solution. We now argue that the algorithm is a $(\gamma+1)$-approximation algorithm. Let $C_1$, $C_2$, $\ldots C_t$ be the components found by the algorithm in this manner. Since $X$ must contain at least one vertex in each $C_i$ it follows that $t \leq |X|$. Thus $\bigcup_{i \leq t} N(C_i) \leq \gamma|X|$. Furthermore for each $C$, $|X_C| \leq |X \cap C|$ and the proof follows.
\end{proof}

We use Lemma~\ref{lem:transversapprox} in conjunction with the following lemma, which is a combination of Lemmata ~3.2 and 3.3 of~ \cite{F.V.Fomin:2010oq}.

\begin{lemma}[ \cite{F.V.Fomin:2010oq}]\label{lem:minor_bidem} Let $\Pi$ be a minor- (contraction-) bidimensional problem with the separation property and $H$ be a (apex) graph. There exists a constant $\eta$ such that for every $G$ excluding   $H$ as a minor, there is a subset $X \subseteq V(G)$  such that $|X|=\cO(\pi(G))$, and $\tw(G \setminus X)\leq \eta$. 
\end{lemma}

Together Lemmata~\ref{lem:transversapprox} and~\ref{lem:minor_bidem} yield the following corollary.

\begin{corollary}\label{cor:minortransvers} 
Let $\Pi$ be a minor- (contraction-) bidimensional problem with the separation property and $H$ be a (apex) graph. There exists a constant $\eta$ such that for every $G$ excluding $H$ as a minor, $\Pi$ is $\eta$-transversable on $H$-minor-free graphs.
\end{corollary}

Finally, combining Theorem~\ref{thm:eptastrans} with Corollaries~\ref{cor:minortreewidth} and~\ref{cor:minortransvers} implies the main theorem of this article. 

\begin{theorem}\label{thm:minoreptas} Let $\Pi$ be a reducible minor- (contraction-) bidimensional problem with the separation property and $H$ be a (apex) graph. There is an EPTAS for $\Pi$ on the class of graphs excluding $H$ as a minor.
\end{theorem}

\section{Applications}

\subsection{Domination and Connectivity Problems}
In {\sc $r$-Dominating Set}, we are given a graph $G$, the objective is to find a minimum size subset $S \subseteq V(G)$ such that $B_G^r(S)=V(G)$. For $r=1$, if we demand that $G[S]$ is connected we obtain the {\sc Connected Dominating Set} problem. In {\sc Connected Vertex Cover} we are given a graph $G$ and the objective is to find a minimum size subset $S \subseteq V(G)$ such that $G[S]$ is connected and every edge in $E(G)$ has at least one endpoint in $S$. It well-known that  {\sc $r$-Dominating Set}, {\sc Connected Dominating Set} and {\sc Connected Vertex Cover} are contraction-bidimensional~\cite{DFHT05}. 

Let $V(G)=L \uplus S \uplus R$ with $N(L) \subseteq S$ and $N(R) \subseteq S$. Let $G_L$ and $G_R$ be defined as in Definition~\ref{def:contrasep}. Let $Z$ be a minimum size $r$-dominating set of $G$ and $Z_L$ be a minimum size $r$-dominating set of $G_L$. We have that  $|Z_L| \leq |Z \setminus R| + |S|$ because $(Z \setminus R) \cup S$ is an $r$-dominating set of $G_L$ and hence $|Z_L| > |Z \setminus R| + |S|$ contradicts the choice of $Z_L$. 
Hence  $|Z_L| \leq |Z \setminus R| + O(|S|)$ and {\sc $r$-Dominating Set} is separable.

We now show that {\sc $r$-Dominating Set} is reducible. Given a graph $G$ and set $X$, let $G' = G \setminus X$ and let $R = B_G^r(X) \setminus X$. Clearly $\tw(G') = \tw(G \setminus X)$. The annotated problem $\Pi'$ is to find a minimum sized set $S' \subseteq V(G')$ such that every vertex in $V(G') \setminus (S \cup R)$ is of distance at most $r$ from a vertex in $S'$. Notice that for any $r$-dominating set $S$ of $G$, $S \setminus X$ is a feasible solution to $\Pi'$ on $G'$. Conversely, for any feasible solution $S'$ of $\Pi'$ on $G'$, we have that $S' \cup X$ is an $r$-dominating set of $G$. Hence {\sc $r$-Dominating Set} is reducible.

The proof that {\sc Connected Dominating Set} and {\sc Connected Vertex Cover} are separable are almost identical to the proof for  {\sc $r$-Dominating Set}  with $r=1$. One only has to note that if $Z$ is an optimal solution to $G$ then $Z \setminus R$ can be made into a feasible solution to $G_L$ by adding $S$ and then observing that $O(|S|)$ vertices are sufficient to connect the resulting connected components.

We now prove that {\sc Connected Dominating Set} is reducible. Given a graph $G$ and set $X$, let $G' = G \setminus X$ and let $R = N(X)$.  The annotated problem $\Pi'$ is to find a minimum sized set $S' \subseteq V(G')$ such that every vertex in $V(G') \setminus (S \cup R)$ has a neighbour in $S'$ and every connected component of $G'[S']$ contains a vertex in $R$. Notice that for any connected dominating set $S$ of $G$, $S \setminus X$ is a feasible solution to $\Pi'$ on $G'$. Conversely, for any feasible solution $S'$ of $\Pi'$ on $G'$, we have that $S = S' \cup X$ is a dominating set of $G$ and has at most $|X|$ connected components. Since $S$ is a dominating set it is sufficient to add $2(|X|-1)$ vertices to $S$ in order to make it a connected dominating set of $G$. Hence {\sc Connected Dominating Set} is reducible. The proof that {\sc Connected Vertex Cover} is reducible is identical.

Finally, let us remark that {\sc Connected Vertex Cover} is $0$-transversable. Given a graph $G$ we find a maximal matching in linear time and output the endpoints of the matching as $X$. Any vertex cover must contain at least one endpoint from each edge in the matching, and thus $|X| \leq 2\pi(G)$. Also, $\tw(G \setminus X) = 0$.

\begin{lemma}\label{lem:domprobs}  {\sc $r$-Dominating Set}, {\sc Connected Dominating Set} and {\sc Connected Vertex Cover} are contraction-bidimensional, separable and reducible. Thus they are $\eta$-transversable on apex-minor-free graphs. Furthermore, {\sc Connected Vertex Cover} is $0$-transversable on general graphs.
\end{lemma}

\paragraph{Max Leaf Spanning Tree.}
In the {\sc Max Leaf Spanning Tree} problem we are given a connected graph $G$ and asked to find a spanning tree $T$ of $G$ maximizing the number of leaves of $T$. We could have shown that the problem is minor-bidimensional and separable, however, just as for {\sc Connected Vertex Cover}, it is easier to show that the problem is $2$-transversable directly. In particular, 
 Kleitman and West \cite{Kleitman:1991kt}  
 have shown that a connected graph which contains no spanning tree with at least $k$ leaves has at most $4k+2$ vertices of degree at least $2$. Thus given a graph we can just return all vertices of degree at least $3$, and the remaining graph will have treewidth at most $2$. Hence,  {\sc Max Leaf Spanning Tree} is $2$-transversable.

We prove that {\sc Max Leaf Spanning Tree} is reducible. Given a graph $G$ and set $X$, let $G' = G \setminus X$ and let $R = N(X)$.  The annotated problem $\Pi'$ is to find a maximum size set $S' \subseteq V(G')$ such that every vertex in $S' \setminus R$ has a neighbour outside of $S'$ and every connected component of $G \setminus S$ contains at least one vertex of $R \setminus S$. For a spanning tree $T$ of $G$ let $S$ be the set of leaves of $T$. Then $S \setminus X$ is a feasible solution to the annotated problem. On the other hand, given a feasible solution $S'$ to $\Pi'$, observe every component of $G \setminus S'$ contains a vertex of $X$. We construct a spanning forest $F$ of $G$ with at most $|X|$ components by picking a spanning tree for every component of $G \setminus S$ and for every vertex $v$ in $S \setminus R$ we connect $v$ to a neighbour outside of $S$. Notice that all vertices of $S$ are leaves of the spanning forest $F$. From $F$ we can construct a spanning tree $T$ by adding at most $|X|-1$ edges. Thus $T$ has at least $|S|-2(|X|-1)$ leaves and we conclude that   {\sc Max Leaf Spanning Tree} is reducible.

\begin{lemma}\label{lem:maxleaf} {\sc Max Leaf Spanning Tree} is $2$-transversable and reducible. \end{lemma}

\subsection{Covering and Packing Problems}

\paragraph{Minor Covering and Packing}
We give below a few generic problems each of which subsumes many problems in itself and fit in our framework. Let ${\cal H}$ be a finite set of connected graphs such that at least one graph in ${\cal H}$ is planar.

\begin{center}
\fbox{\begin{minipage}{13cm}
\noindent {\sc Vertex-${\cal H}$-Covering}\\ 
{\sl Input:} A graph $G$\\
{\sl Objective:} Find a minimum size set $S\subseteq V(G)$ such that $G \setminus S$ does not \\ 
\phantom{Objective:} contain any of the graphs from ${\cal H}$ as a minor.
\end{minipage}}
\end{center}
\noindent

\begin{center}
\fbox{\begin{minipage}{13cm}
\noindent {\sc Vertex-${\cal H}$-Packing}\\
{\sl Input:} A graph $G$.\\
{\sl Objective:} Find a maximum size collection of vertex disjoint subgraphs \\
\phantom{Objective:} $G_{1},\ldots,G_{k}$ of $G$ such that each of them contains some graph \\
\phantom{Objective:} from ${\cal H}$ as a minor.
\end{minipage}}
\end{center}
\noindent

It it easy to see that both {\sc Vertex-${\cal H}$-Covering} and {\sc Vertex-${\cal H}$-Packing} are minor-closed problems. Now, let $h$ be the size of the smallest planar graph $H$ in ${\cal H}$. By a result of Robertson et al.~\cite{RobertsonST94}, $H$ is a minor of the $(t \times t)$-grid, where $t = 14|V(H)|-24$. Consider a $(r \times r)$-grid $F$. $F$ contains $r^2/t^2$ disjoint $H$-minors. Any covering of $F$ must pick at least one vertex from each of these minors, therefore both {\sc Vertex-${\cal H}$-Covering} and {\sc Vertex-${\cal H}$-Packing} are minor-bidimensional.

We now prove that {\sc Vertex-${\cal H}$-Covering} is separable. Given a graph $G$ and a partition of $V(G)$ into $L$, $S$ and $R$ such that $N(L) \subseteq S$ and $N(R) \subseteq S$, let $Z$ be a set of minimum size such that $G \setminus Z$ contains no graph of ${\cal H}$ as a minor. Consider the smallest set $Z_L$ such that $G[L] \setminus Z_L$ contains no graph of ${\cal H}$ as a minor. If $|Z \cap L| < |Z_L|$, this contradicts the choice of $Z_L$ since $G[L] \setminus (Z \cap L)$ does not contain a graph of ${\cal H}$ as a minor. 
The proof for $Z_R$ is identical, thus {\sc Vertex-${\cal H}$-Covering} is separable.

The proof of separability of {\sc Vertex-${\cal H}$-Packing} goes along the same lines as the proof for  {\sc Vertex-${\cal H}$-Covering}, but with a few notable differences. In particular, we formalize {\sc Vertex-${\cal H}$-Packing} as a graph optimization problem where we seek an edge set $Z \subseteq E(G)$. The objective function, $\kappa_{COV}$, {\em counts} the number of connected components of the subgraph formed by $Z$ 
  that contain { at least one} copy of a graph in ${\cal H}$ as a minor, and all edge subsets are considered feasible.

Given a graph $G$ and a partition of $V(G)$ into $L$, $S$ and $R$ such that $N(L) \subseteq S$ and $N(R) \subseteq S$, let $Z$ be an edge set maximizing $\kappa_{COV}(G,Z)$ and $Z_L$ be an edge set maximizing $\kappa_{COV}(G[L],Z_L)$. By the choice of $Z_L$ we have $\kappa_{COV}(G[L],Z_L) \geq \kappa_{COV}(G[L],Z \cup E(G[L]))$. 
The proof for $Z_R$ is identical, hence {\sc Vertex-${\cal H}$-Packing} is separable.

Finally, it is easy to see that both {\sc Vertex-${\cal H}$-Covering} and {\sc Vertex-${\cal H}$-Packing} are reducible. Given   $G$ and $X$ we let $G'=G \setminus X$. For {\sc Vertex-${\cal H}$-Covering} $X$ can be added to the an optimal solution in $G'$ at the cost of $|X|$. For {\sc Vertex-${\cal H}$-Packing} at most $|X|$ of the minors of graphs in ${\cal H}$ contained a vertex in $X$ and got removed when $X$ was deleted. Expressing both problems as \pmm \ problems is routine. 

\begin{lemma}\label{lem:minorpackcover} {\sc Vertex-${\cal H}$-Covering} and {\sc Vertex-${\cal H}$-Packing} are minor-bidimensional, separable and reducible. 
\end{lemma}

{\sc Vertex-${\cal H}$-Covering} contains various problems as a special case, for example: (a) {\sc Vertex Cover} by letting ${\cal H}$ contain a single graph on a single edge, (b) {\sc Feedback Vertex Set} by setting ${\cal H}$ to contain a single graph, the complete graph on $3$ vertices $K_3$; (c) {\sc Diamond Hitting Set} by letting ${\cal H}$ contain a single graph, the comlete graph on 4 vertices $K_4$ minus a single edge. Other choices for ${\cal H}$ lead to vertex deletion into outerplanar graphs, series-parallell graphs, graphs of constant treewidth ({\sc $\eta$-Transversal}) or pathwidth. On the other hand, {\sc Vertex-${\cal H}$-Packing} contains problems like {\sc Cycle Packing} as a special case. 

\paragraph{Subgraph Covering and Packing} 
Now we consider problems about covering and packing subgraphs. These problems can be handled in much the same way as covering and packing minors. Let ${\cal S}$ be a finite set of connected graphs.

\begin{center}
\fbox{\begin{minipage}{13cm}
\noindent {\sc Vertex-${\cal S}$-Covering}\\ 
{\sl Input:} A graph $G$\\
{\sl Objective:} Find a minimum size set $S\subseteq V(G)$ such that $G \setminus S$ does not \\ 
\phantom{Objective:} contain any of the graphs from ${\cal S}$ as a subgraph.
\end{minipage}}
\end{center}
\noindent

\begin{center}
\fbox{\begin{minipage}{13cm}
\noindent {\sc Vertex-${\cal S}$-Packing}\\
{\sl Input:} A graph $G$.\\
{\sl Objective:} Find a maximum size collection of vertex disjoint subgraphs \\
\phantom{Objective:} $G_{1},\ldots,G_{k}$ of $G$ such that each of them contains some graph \\
\phantom{Objective:} from ${\cal S}$ as a subgraph.
\end{minipage}}
\end{center}
\noindent

We will not show that {\sc Vertex-$\cal S$-Covering} or {\sc Vertex-$\cal S$-Packing} are bidimensional. Instead, we will give a reduction rule, and show that instances reduced according to this rule have an $r$-dominating set of size $O(OPT)$, where $r$ is the maximum size of a graph in ${\cal S}$. Since $r$-{\sc Dominating Set} is $\eta$-transversable there is an algorithm that in polynomial time outputs a set $X \subseteq V(G)$ of size $O(OPT)$ such that $\tw(G \setminus X) \leq \eta$. Hence the pre-processed version of {\sc Vertex-$\cal S$-Covering} and {\sc Vertex-$\cal S$-Packing} is $\eta$-transversable. 

Consider the following rule, the {\em Redundant Vertex Rule}. Given as input $G$ to  {\sc Vertex-$\cal S$-Covering} or {\sc Vertex-$\cal S$-Packing} remove all vertices that are not part of any subgraph isomorphic to any graph in ${\cal S}$. We can perform the Redundant Vertex Rule in $O(|V|\cdot |{\cal S}|)$  time by looking at a small ball around every vertex $v$ and check whether the ball contains a subgraph isomorphic to a graph in $\cal S$ that contains $v$. This algorithm to check a subgraph isomorphic to a given graph containing a particular vertex appears in an article of Eppstein~\cite{Eppstein99}. 

Consider an instance $G$ of {\sc Vertex-$\cal S$-Covering} reduced according to the Redundant Vertex Rule, and let $S$ be an optimal solution to $G$. Since $X$ hits all copies of graphs in ${\cal S}$ occuring in $G$ and every vertex in $G$ appears in some copy of a graph in ${\cal S}$ it follows that $X$ is a $r$-dominating set of $G$, where $r$ is the maximum size of a graph in ${\cal S}$. Finally, consider an instance $G$ of {\sc Vertex-$\cal S$-Packing} reduced according to the Redundant Vertex Rule, and consider an optimal solution $G_{1},\ldots,G_{OPT}$ such that for every $i$, $G_i$ contains a graph
from ${\cal S}$ as a subgraph. Since every vertex occurs a subgraph of $G$ isomorphic to a graph in $H$, the selection of $G_{1},\ldots,G_{OPT}$ implies that every vertex $v$ has distance at most $r$ to some vertex in some $G_i$. Let $X=\{v_1,v_2,\ldots,v_{OPT}\}$ where $x_i \in V(G_i)$. Then every vertex $v$ has distance at most $2r$ to $X$. Thus, Lemma~\ref{lem:domprobs} yields that {\sc Vertex-$\cal S$-Covering} or {\sc Vertex-$\cal S$-Packing} are $\eta$-transversable. The proof that they are both reducible is identical to the discussion for {\sc Vertex-$\cal H$-Covering} or {\sc Vertex-$\cal H$-Packing}.

\begin{lemma}\label{lem:subgraphcover}  {\sc Vertex-$\cal S$-Covering} or {\sc Vertex-$\cal S$-Packing} pre-processed with the Redundant Vertex Rule are $\eta$-transversable and reducible.
\end{lemma}

\subsection{Partial Domination and Covering}
In the {\sc Partial $r$-Dominating Set} problem we are given a graph $G$ together with an integer $t \leq |V(G)|$. The objective is to find a minimum size set $S$ such that $|B_G^r(S)| \geq t$. In {\sc Partial Vertex Cover} we are given a graph $G$ together with an integer $t \leq |E(G)|$ and the objective is to find a minimum size vertex set $S$ such that $|\{uv \in E: u \in S \vee v \in S\}| \geq t$. We will call $\{uv \in E: u \in S \vee v \in S\}$ the set of edges {\em covered} by $S$. PTAS for {\sc Partial Vertex Cover} on planar graphs was given in 
\cite{Gandhi2004}. We are not aware of any PTAS  for  {\sc Partial $r$-Dominating Set}.

We will not show that  {\sc Partial $r$-Dominating Set} and {\sc Partial Vertex Cover} are bidimensional, instead we will directly construct a EPTASs for these problems on apex-minor-free graphs using the tools developed so far. We will use $OPT$ for the size of an optimal solution to our instances. Let $H$ be a fixed apex graph, our input graphs will exclude $H$ as a minor. We  employ an algorithm of Fomin et al.~\cite{FominLRS09}. They give an algorithm for solving  {\sc Partial $r$-Dominating Set} in time $2^{O(r\sqrt{OPT})}n^{O(1)}$ and {\sc Partial Vertex Cover} in time $2^{O(\sqrt{OPT})}n^{O(1)}$. A key part of their algorithm for {\sc Partial $r$-Dominating Set} is a polynomial time algorithm (\cite{FominLRS09}, Lemma $5$) that given a graph $G$ together with integers $t$ and $k$ returns an induced subgraph $G'$ of $G$ such that $G$ has a $k$-sized vertex set $S$ such that $|B_G^r(S)| \geq t$ if and only if $G'$ has a $k$-sized vertex set $S'$ such that $|B_{G'}^r(S')| \geq t$. Furthermore, $G'$ has a $3r$-dominating set of size $k$. Our EPTAS loops over all possible values of $k$ and for each such value produces $G_k'$ from $G$, $t$ and $k$ using Lemma $5$ of~\cite{FominLRS09}. If $G_k'$ has less than $t$ vertices, then $G_k'$ cannot have any set which covers at least $t$ vertices, and so, neither can $G$. If $G'_k$ has at least $t$ vertices, we proceed with the following subroutine.

Just as in the proof of Theorem~\ref{thm:eptastrans}, let $\epsilon'$ be a constant to be chosen later. By construction $G'_k$ has $3r$-dominating set of size $k$. Since $3r$-{\sc Dominating Set} is $\eta$-transversable there is a polynomial time algorithm that outputs a set $X$ of size at most $\rho k$ such that $\tw(G'_k \setminus X) \leq \eta$. By Lemma~\ref{lem:theLemma} there is a polynomial time algorithm that computes a set $X'$ of size at most $\epsilon' \rho k$ such that $\tw(G'_k \setminus X') \leq \delta$ for a constant $\delta$ depending only on $\eta$ and $H$. We put all vertices in $X'$ in our solution. Specifically, we remove $X'$ from $G'_k$ and put all other vertices of $B_{G'_k}^r(X')$ into a set $R$. Using standard dynamic programming (or by formulating the problem in an extended version of MSO~\cite{ArnborgLS91}) on graphs of bounded treewidth, we can find a minimum size set $S' \subseteq V(G'_k) \setminus X'$ such that $|X'| + |R \cup B_{G'_k \setminus X'}^r(S')| \geq t$ in time $f(\delta)n^{O(1)}$. The subroutine returns the set $S' \cup X'$ as a solution.

Since $G'$ is an induced subgraph of $G$, any solution $S = S' \cup X'$ returned by the subroutine covers at least $t$ vertices in $G$. We return the smallest $S$ as our $(1+\epsilon)$-approximate solution. In the iteration of the outer loop where $k=OPT$ we have that $G'_k$ has a set $Z$ of size $OPT$ that covers $t$ vertices in $G'$. Observe that $Z \setminus X'$ covers at least $t - |B_{G'_k}^r(X')|$ of $V(G'_k) \setminus B_{G'_k}^r(X')$ in the graph $G'_k \setminus X'$. Thus the solution returned by the dynamic programming algorithm has size at most $|Z \setminus X'| \leq |Z| = OPT$ and the solution returned by the subroutine in this iteration is at most $OPT+|X'| \leq OPT(1+\epsilon'\rho)$. Chosing $\epsilon' = \epsilon/\rho$ concludes the analysis of our EPTAS for {\sc Partial $r$-Dominating Set}. An EPTAS for {\sc Partial Vertex Cover} can be constructed in a similar manner.

\begin{lemma}
\label{lem:partial}
There is an EPTAS for {\sc Partial $r$-Dominating Set} and  {\sc Partial Vertex Cover} on apex-minor-free graphs. 
\end{lemma}

Finally by applying Theorems~\ref{thm:eptastrans} and \ref{thm:minoreptas} together with Lemmata~\ref{lem:domprobs}, \ref{lem:maxleaf}, \ref{lem:minorpackcover}, \ref{lem:subgraphcover} and \ref{lem:partial} we get the following 
corollary. 
\begin{corollary}
\label{metacor1} 
{\sc Feedback Vertex Set},  {\sc Vertex Cover}, {\sc Connected Vertex Cover}, {\sc Cycle Packing}, {\sc Diamond Hitting Set}, 
{\sc Minimum-Vertex Feedback Edge Set}, {\sc Vertex-${\cal H}$-Packing}, 
{\sc Vertex-${\cal H}$-Covering}, {\sc Maximum Induced Forest}, {\sc Maximum Induced Bipartite Subgraph} and {\sc Maximum Induced Planar Subgraph}
admit an EPTAS on $H$-minor-free graphs.
{\sc Edge Dominating Set}, {\sc Dominating Set}, {\sc $r$-Dominating Set}, {\sc $q$-Threshold Dominating Set},  
{\sc Connected Dominating Set}, {\sc Directed Domination},  {\sc $r$-Scattered Set}, 
{\sc Minimum Maximal Matching}, {\sc Independent Set},  {\sc Maximum Full-Degree Spanning Tree}, 
{\sc Max Induced at most $d$-Degree Subgraph}, 
{\sc Max Internal Spanning Tree}, 
{\sc Induced Matching}, {\sc Triangle Packing}, {\sc Vertex-${\cal S}$-Covering},  
{\sc Vertex-${\cal S}$-Packing}  {\sc Partial $r$-Dominating Set} and  {\sc Partial Vertex Cover} admit an  EPTAS on apex-minor-free graphs.
\end{corollary}

It should be noted that for a fixed $\epsilon$,  the treewidth $\tau$ in Corollary~\ref{cor:theCorollary} is $O(1/\epsilon)$. For $H$-minor-free graphs one can compute the set $X'$ from $G$ and $X$ using the procedure described in the last paragraph of the proof of Lemma~\ref{lem:theLemma} but using the constant factor approximation for treewidth on $H$-minor-free graphs~\cite{FeigeHajLee05} instead of the approximation algorithm for general graphs. For many problems discussed in this paper, the MSO-based algorithms on graphs of bounded treewidth could be replaced by standard dynamic programming algorithms with running time $2^{O(\tw(G))}n$ or $2^{O(\tw(G)\log (\tw(G)))}n$. On $H$-minor-free graphs this leads to EPTASs with running times on the form $2^{O(1/\epsilon)}n + n^{O(1)}$ or $2^{O(1/\epsilon \log (1/\epsilon))}n + n^{O(1)}$, which is comparable to the fastest previously known results.

\bibliographystyle{siam}
\bibliography{graphs,kernels,convcNK,complete}

\end{document}